\documentclass[journal]{IEEEtran}
\usepackage{cite}
\usepackage{amsmath,amssymb,amsfonts,amsthm}
\usepackage{algorithmic} 
\usepackage{graphicx}
\usepackage{tikz}
\usepackage{graphicx}
\usepackage{pgfplots}
\usepackage{mathrsfs}  
\usepackage{caption}
\usepackage{subcaption}
\usepackage{comment}
\usepackage{float}
\usetikzlibrary{spy,backgrounds}
\usepackage{lipsum}
\usepackage{siunitx}
\usepackage{caption}
\newtheorem{theorem}{Theorem}
\newtheorem{lemma}{Lemma}

\newtheorem{definition}{Definition}

\usepackage{tabularx}
\usepackage{booktabs}  
\usepackage{caption}
\theoremstyle{remark}
\newtheorem{remark}{Remark}

\usepackage{tikz}
\usepackage{caption}
\usepackage{float}
\usetikzlibrary{decorations.pathreplacing}

\newcommand\numeq[1]%
  {\stackrel{\scriptscriptstyle(\mkern-1.5mu#1\mkern-1.5mu)}{=}}
  \newcommand\numl[1]%
  {\stackrel{\scriptscriptstyle(\mkern-1.5mu#1\mkern-1.5mu)}{<}}
\newcommand\numleq[1]%
  {\stackrel{\scriptscriptstyle(\mkern-1.5mu#1\mkern-1.5mu)}{\leq}}
\newcommand\numgeq[1]%
  {\stackrel{\scriptscriptstyle(\mkern-1.5mu#1\mkern-1.5mu)}{\geq}}
\usepackage{mathtools}

\DeclarePairedDelimiterX\braket[2]{\langle}{\rangle}{#1 \delimsize\vert #2}

\usepackage{breqn}

\usepackage{xcolor}
\def\BibTeX{{\rm B\kern-.05em{\sc i\kern-.025em b}\kern-.08em
    T\kern-.1667em\lower.7ex\hbox{E}\kern-.125emX}}
    
    \usepackage{cuted}
\usepackage{lipsum}

\ifCLASSINFOpdf 
\else  
\fi

\begin{document}

\title{Finite-Length Empirical Comparison of Polar, PAC, and Invertible-Extractor Secrecy Codes over the Wiretap BSC}
\author{\IEEEauthorblockN{Jaswanthi Mandalapu\IEEEauthorrefmark{1}, Tyler Kann\IEEEauthorrefmark{2}, Andrew Thangaraj\IEEEauthorrefmark{1}, Matthieu R. Bloch\IEEEauthorrefmark{2}\\  ee19d700@smail.iitm.ac.in, kann@gatech.edu,  andrew@ee.iitm.ac.in, matthieu.bloch@ece.gatech.edu  \\}
\IEEEauthorblockA{ \IEEEauthorrefmark{1} Department of Electrical Engineering, Indian Institute of Technology Madras\\
\IEEEauthorrefmark{2} School of Electrical and Computer Engineering, Georgia Institute of Technology} 
}
\maketitle

\begin{abstract}
We compare three secrecy-coding schemes for the degraded wiretap binary symmetric channel (BSC) in the \emph{finite-blocklength} regime: (i) polar wiretap coset codes, (ii) PAC codes used as wiretap coset codes, and (iii) the invertible-extractor (IE) framework of Bellare--Tessaro.
Our comparison is empirical and uses a common \emph{semantic-secrecy} metric (distinguishing advantage). For polar coset codes, we compute Eve's polarized bit-channel capacities (via the Tal--Vardy construction) to obtain explicit finite-length upper bounds on mutual-information leakage, yielding strong secrecy bounds. For PAC coset codes, we prove that Eve's synthesized bit-channels are equivalent to those of polar codes (up to a permutation), so the same leakage bounds apply; we then convert these strong-secrecy bounds into semantic-secrecy guarantees for symmetric wiretap channels. For the IE scheme, we use the closed-form semantic-secrecy bounds already established.
Finally, we report finite-length results that jointly characterize (a) semantic-secrecy guarantees against Eve and (b) frame-error-rate performance at Bob. Our simulations indicate that Polar and PAC codes exhibit comparable reliability performance under the considered finite-length settings, while PAC codes retain the same secrecy bounds as Polar codes. More broadly, our results highlight that non-universal designs of wiretap codes provide tighter secrecy guarantees than universal design based on invertible extractors.
\end{abstract}

\IEEEpeerreviewmaketitle

\section{Introduction}
Information-theoretic security over noisy channels is classically captured by Wyner's wiretap channel model \cite{wyner1975wire}. While asymptotic secrecy-capacity results are well established, practical deployments typically operate at short to moderate blocklengths where both reliability and secrecy must be assessed \emph{quantitatively} at finite $N$.

This paper has one goal: a \emph{finite-length, simulation-driven} comparison of three secrecy coding schemes for the degraded wiretap BSC:
\begin{enumerate}
    \item polar wiretap coset codes (Mahdavifar--Vardy) \cite{Hessam-Vardy},
    \item PAC codes used as wiretap coset codes (Arıkan) \cite{arikan2019sequential},
    \item invertible-extractor (IE) based schemes (Bellare--Tessaro) \cite{bellare2012polynomial,bellare2012cryptographic}.
\end{enumerate}

A key challenge is that these schemes are traditionally evaluated under different secrecy notions. Polar/PAC coset codes are typically analyzed via \emph{strong secrecy} (mutual-information leakage for uniformly distributed messages), whereas IE-based schemes are designed and analyzed directly under \emph{semantic secrecy} (distinguishing advantage, worst-case over message distributions). For binary-input symmetric wiretap channels (including the wiretap BSC), strong secrecy implies semantic secrecy, e.g.,
$\delta^{ds}\le \sqrt{2I(\mathbf{M};\mathbf{Z})}$ for coset coding \cite{bellare2012cryptographic}. This lets us put all three schemes on the same semantic-secrecy scale.

Our contributions are:
\begin{itemize}
    \item A unified finite-length evaluation pipeline that reports Bob FER and Eve semantic-secrecy advantage for all three schemes.
    \item A proof that, from Eve's perspective, PAC codes preserve the synthesized bit-channels of polar codes (up to a permutation). Consequently, the standard polar leakage bounds apply unchanged to PAC coset codes.
    \item Finite-length numerical comparisons of Polar and PAC codes in terms of Bob's reliability, together with a comparison of their semantic-secrecy guarantees against those of the IE framework.
    \item Perhaps surprisingly, under the finite-length bounds considered here, polar/PAC secrecy codes provide tighter secrecy guarantees than the IE framework in the sense that the resulting upper bounds on Eve's information leakage (and hence on semantic-secrecy advantage) are smaller.
\end{itemize}

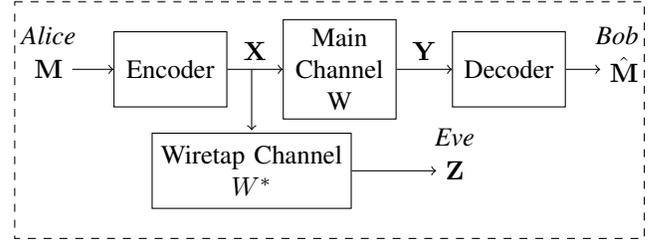
\begin{figure}[t!]
    \centering
    \begin{tikzpicture}[scale = 0.9]
  \node at (0,0) (u) {$\mathbf{M}$};
  \node[draw, minimum width=1.5cm, minimum height=1cm] at (1.8,0) (enc) {Encoder};
  \node[draw, minimum width=1.5cm, minimum height=1cm, align = center] at (4.3,0) (main) {Main \\ Channel \\ W};
  \node[draw, minimum width=1.5cm, minimum height=1cm] at (6.8,0) (dec) {Decoder};
  \node at (8.5,0) (uhat) {$\hat{\mathbf{M}}$};

  \node[draw, minimum width=1.5cm, minimum height=1cm, align=center] at (3,-1.5) (wt) {Wiretap Channel\\ $W^*$};
  \node at (6,-1.5) (z) {$\mathbf{Z}$};

  \node at (0, 0.5) {\textit{Alice}};
  \node at (8.4, 0.5) {\textit{Bob}};
  \node at (6, -1) {\textit{Eve}};

  \draw[->] (u) -- (enc);
  \draw[->] (enc) -- node[above] {$\mathbf{X}$} (main);
  \draw[->] (main) -- node[above] {$\mathbf{Y}$} (dec);
  \draw[->] (dec) -- (uhat);
  \draw[->] (3,0) -- (3,-0.9);
  \draw[->] (wt) -- (z);
  \draw[dashed] (-0.5,1) -- (8.8,1) -- (8.8,-2.5) -- (-0.5,-2.5) -- (-0.5,1);

\end{tikzpicture}
\caption{The Wiretap Channel Setting}
    \label{Fig-1:wiretap}
\end{figure}

\section{System Model and Comparison Framework}
\subsection{A Wiretap BSC Channel}

We first briefly review the wiretap binary symmetric channel (BSC), denoted by $\mathsf{BSC}(p_b,p_e),$ with ${0<p_b<p_e<1}$, as depicted in Fig.~\ref{Fig-1:wiretap}. The main channel between Alice (transmitter) and Bob (receiver) is a BSC with crossover probability $p_b$, represented by the triple $\langle\mathcal{X},\mathcal{Y},W\rangle$, where $\mathcal{X}=\mathcal{Y}=\{0,1\},$ and $W = \left[\begin{array}{cc}
          1-p_b & p_b \\ p_b & 1-p_b 
\end{array}\right].$ The eavesdropper Eve observes a \emph{degraded} version of the main channel, $\langle\mathcal{X},\mathcal{Z},W^{*}\rangle$, with ${\mathcal{Z}=\{0,1\}}$ and $W^{*}=\mathsf{BSC}(p_e)$, i.e., $W^{*} = \left[\begin{array}{cc}
          1-p_e & p_e \\ p_e & 1-p_e 
\end{array}\right].$

Precisely, a $(k,N)$ wiretap code over the $\mathsf{BSC}(p_b,p_e)$ consists of the following components:
\begin{itemize}
    \item An encoder $\mathcal{E}$ mapping a $k$-bit secret message ${\mathbf{M}=(M_1,\ldots,M_k)}$ from Alice to an $N$-bit codeword ${\mathbf{X}=(X_1,\ldots,X_N)}$.
    \item A decoder $\mathcal{D}$ enabling Bob to estimate $\mathbf{M}$ from $\mathbf{Y}=(Y_1,\ldots,Y_N)$ obtained via $N$ independent uses of $W$.
    \item An eavesdropper observing $\mathbf{Z}=(Z_1,\ldots,Z_N)$ through $N$ independent uses of $W^*$.
\end{itemize}
The secrecy rate in this setup is defined as $R_s = \frac{k}{N}.$ Furthermore, from \cite{1055763}, the following is known about the secrecy capacity (max rate of transmission) of the wiretap $\mathsf{BSC}(p_b,p_e).$
\begin{theorem}
    For uniformly distributed input information bits, the secrecy capacity of $\mathsf{BSC}(p_b,p_e)$ is
    $$
        C_s = h_2(p_e) - h_2(p_b) ~\text{bits per channel use,}
    $$
    where $h_2(\cdot)$ denotes the binary entropy function.
\end{theorem}
The well-known coding objective here is to achieve $C_s$ while guaranteeing both reliability for Bob and \emph{strong secrecy} against Eve. Formally, a sequence of codes must satisfy:
\begin{itemize}
\item Reliability:
\begin{equation}\label{eq-1:reliabilitycond}
\lim_{k\to\infty}\Pr({\mathbf{M}\neq \hat{\mathbf{M}}})=0.
\end{equation}
\item Strong Secrecy:
\begin{equation}\label{eq-2:securitycond}
\lim_{k\to\infty} I(\mathbf{M};\mathbf{Z})=0,
\end{equation}
\end{itemize}
where $I(\mathbf{M};\mathbf{Z})$ denotes the mutual information secrecy (MIS) leakage to EveIn this paper we focus on \emph{finite} blocklengths and compare secrecy codes using two quantities: (i) Bob's frame error rate (FER) under practical decoding, and (ii) Eve's \emph{semantic-secrecy} advantage, denoted by $\delta^{ds}$ following \cite{bellare2012cryptographic}. For binary-input symmetric wiretap channels, strong secrecy bounds based on $I(\mathbf{M};\mathbf{Z})$ can be converted into semantic-secrecy bounds; we use this conversion to put wiretap coset codes (polar and PAC) and the invertible-extractor scheme on the same semantic-secrecy scale.

We next summarize the three finite-length constructions studied in this work and the bounds we use to evaluate them.

\subsection{Secrecy Codes and Their Finite-Length Bounds}
For our finite-length analysis, we consider two secrecy coding schemes: (i) the strong secrecy construction of \cite{Hessam-Vardy}, and (ii) the invertible extractor framework of \cite{bellare2012cryptographic}. A brief review of these schemes is given below.

Let $\mathcal{C}(\cdot)$ denote the underlying error-correcting encoder. 

\emph{Strong Secrecy Codes \cite{Hessam-Vardy}:} In this approach, a codeword of blocklength $N$ is partitioned into three subsets: $\mathcal{A}$ (information bits), $\mathcal{R}$ (random bits), and $\mathcal{B}$ (frozen bits). Let $|\mathcal{A}| = k$ and $|\mathcal{R}| = r$, which implies $|\mathcal{B}| = N - k - r$. Then,
\begin{definition}\label{defn-secrecycodes1}
The encoder for a secrecy code of length $N$ associated with $(\mathcal{A}, \mathcal{R})$ is a function
$
    \mathcal{E}:\{0,1\}^k \times \{0,1\}^r \to \{0,1\}^N.
$
It takes as input a secret message $\mathbf{M} \in \{0,1\}^k$ and a random vector $\mathbf{R} \in \{0,1\}^r$, whose entries are drawn independently and uniformly at random. The encoder constructs an $N$-bit vector $\mathbf{V}$ by setting $\mathbf{V}_{\mathcal{A}} = \mathbf{M}$, $\mathbf{V}_{\mathcal{R}} = \mathbf{R}$, and $\mathbf{V}_{\mathcal{B}} = 0$. The final codeword $\mathbf{X}$ is obtained as
$\mathcal{E}(\mathbf{M},\mathbf{R}) = \mathcal{C}(\mathbf{V}).$
\end{definition}
When $\mathcal{C}(\cdot)$ is chosen to be a Polar code, the following upper bound holds on Eve’s information leakage \cite{Hessam-Vardy}.
\begin{theorem}\label{thm:secrecy}
    [The Equivocation Bound] The MIS leakage of polar secrecy codes associated with $(\mathcal{A},\mathcal{R})$ satisfies
    \begin{align}\label{eq:secrecy}
        I(\mathbf{M};\mathbf{Z}) \leq \sum_{i \in \mathcal{R}^c} C(W_i),
    \end{align}
    where $W_i$ is the $i$-th polarized bit-channel corresponding to the main channel $W$.
\end{theorem}

\emph{Invertible Extractor Codes\cite{bellare2012polynomial}:}
In this scheme, the message $\mathbf{M}$ of length $m$ is divided into $t$ blocks, each of length $b = (h2(p_e) - h2(p_b))N,$
where $p_b$ and $p_e$ are the crossover probabilities of the main and eavesdropper channels, respectively. Then,

\begin{definition}\label{defn:secrecycodes-2}
The encoder divides $\mathbf{M}$ into $t$ blocks $M[1], M[2], \ldots, M[t],$ each of length $b$. It selects a random $k$-bit string $A$ along with independent random strings $R[1], R[2], \ldots, R[t]$. Next, the codeword $\mathbf{X}$ is constructed as follows:
\begin{align*}
\mathcal{E}(\mathbf{M}) = \mathcal{C}(A) || \mathcal{C}(A \odot (M[1]|R[1])) || \ldots|| \mathcal{C}(A \odot (M[t]|R[t])),
\end{align*}
where $\odot$ denotes multiplication of $k$-bit strings interpreted as elements of the extension field $\mathrm{GF}(2^k)$, and $||$ denotes concatenation. Here, $\mathcal{C}$ can be any linear block-code including Polar/PAC codes.
\end{definition}

For this construction, the information leakage to Eve under any general error-correcting code $\mathcal{C}(\cdot)$ is bounded as follows.

\begin{theorem}\label{thm-ie}
Let $\delta^{ds}$ denote the distinguishing security advantage of the scheme; see \cite{bellare2012polynomial} for the definitions. Then, we have 
$\delta^{ds} (\text{extractor}) \leq 6.2^{-\sqrt{N}}.$  
Further, using \cite{bellare2012cryptographic}[Lemma~4.9], the Eve's mutual information leakage is bounded by
\begin{align}\label{eq:ie}
    I(\mathbf{M};\mathbf{Z}) \leq 2\delta^{ds}\log\left( \frac{2^N}{\delta^{ds}}\right).
\end{align}
\end{theorem}

\subsection{Putting All Schemes on a Semantic-Secrecy Metric}
The finite-length secrecy study for wiretap coset codes (such as the polar/PAC secrecy codes) uses the notion of strong secrecy assuming that the message $\mathbf{M}$ is uniformly random to bound Eve's knowledge. However, the notion of distinguishing or semantic secrecy used for studying invertible extractor based secrecy codes maximizes Eve's knowledge over all message distributions. 

In \cite[Section 4.6]{bellare2012cryptographic}, strong secrecy with uniformly random messages is shown to be equivalent to semantic secrecy for binary-input, binary-vector-output symmetric wiretap channels. The same proof can be readily extended to the case of binary-input, symmetric wiretap channels with no restriction on the output alphabet size. A proof sketch is included in the appendix for clarity and completion.

So, wiretap coset codes with strong secrecy are also semantically secure when the wiretap channel is binary-input and symmetric. In particular, using \cite[Theorem 4.5]{bellare2012cryptographic}, we have 
\begin{equation}
    \delta^{ds}(\text{coset code}) \le \sqrt{2I(\mathbf{M};\mathbf{Z})},
\end{equation}
where $\mathbf{M}$ is a uniformly random message in a wiretap coset code used over a binary-input, symmetric wiretap channel with output $\mathbf{Z}$.

In what follows, we analyze and compare Polar and PAC codes when used as wiretap coset codes versus invertible extractors, along with their reliability guarantees, in the \emph{finite-blocklength} regime. Before presenting our empirical analysis, we highlight a key theoretical insight: the secrecy bound in \eqref{eq:secrecy} remains unchanged when the underlying error-correcting code is replaced by a PAC code. This invariance follows from the result that, under the secrecy coding schemes considered, the bit-channels induced by PAC and Polar codes---when observed between Alice and Eve---are equivalent in the information-theoretic sense.

The next theorem summarizes the key property used in our finite-length comparison.

\begin{theorem}[Polar--PAC equivalence for Eve]\label{thm:equiv_polar_pac}
Fix a blocklength $N$ and rate-profile (information/frozen positions). Consider a polar code and the corresponding PAC code that uses the same polar transform $G_N$ with an invertible convolutional precoder $T_N^g$. For any binary-input memoryless channel observed by Eve, the synthesized bit-channels seen by Eve under the polar and PAC constructions are equivalent up to a column permutation. In particular, the bit-channel mutual informations (and hence the leakage upper bounds computed from bit-channel capacities) are identical for polar and PAC coset codes.
\end{theorem}
\begin{proof}
A detailed proof, based on channel symmetry and an explicit permutation argument, is provided in Appendix~\ref{app:equiv}.
\end{proof}

\section{Finite-Length Empirical Comparison}
\subsection{Simulation setup and evaluation methodology}
We consider a degraded wiretap channel $\mathsf{BSC}(p_b,p_e)$ with $p_b<p_e<\tfrac12$. For each blocklength $N$, we report two finite-length metrics:
(i) Bob's FER under successive-cancellation list (SCL) decoding, and (ii) Eve's semantic-secrecy advantage $\delta^{ds}$.
For polar and PAC wiretap coset codes (Defn.~\ref{defn-secrecycodes1}), we upper bound $I(\mathbf{M};\mathbf{Z})$ using the bit-channel-capacity sum in \eqref{eq:secrecy} (computed via the Tal--Vardy construction), and then convert it to a semantic-secrecy bound using $\delta^{ds}\le \sqrt{2I(\mathbf{M};\mathbf{Z})}$ for symmetric wiretap channels. For the IE scheme (Defn.~\ref{defn:secrecycodes-2}), we compute $\delta^{ds}$ directly using the bounds in \cite{bellare2012polynomial,bellare2012cryptographic}.

\begin{remark}
For the secrecy bounds, Polar and PAC coset codes are identical: the PAC convolutional precoder does not change Eve's synthesized bit-channels (up to a permutation), so the leakage upper bounds computed from \eqref{eq:secrecy} coincide. We separately compare their reliability performance at Bob through finite-length simulations, where Polar and PAC codes are observed to exhibit comparable FER performance under the considered settings.
\end{remark}

\subsection{Semantic secrecy comparison of polar, PAC, and IE wiretap coset codes}
In Table~\ref{tab:performance}, we report the finite-length secrecy rates achieved by Polar and PAC codes over the wiretap channel under the polar secrecy construction of Defn.~\ref{defn-secrecycodes1}. We consider a main channel with cross over probability $p_b = 0.05$ and explicitly compute the upper bound on Eve's information leakage, denoted by $\bar{I}$, as characterized in Theorem~\ref{thm:secrecy}. The evaluation is performed across different values of Eve’s crossover probability $p_e$ and varying message lengths $k$, while maintaining that Bob’s Frame Error Rate (FER) within the range $[0.05, 0.06]$. For each configuration, we report both the nominal secrecy rate $k/N$, which ignores Eve’s information leakage, and the effective secrecy rate $R_{\mathrm{eff}} = (k-\bar{I})/N$, which accounts for the finite-length leakage bound. In addition, we present the resulting semantic-secrecy guarantee using the bound $\delta^{ds}(polar) \leq \sqrt{2\bar{I}}$. As a direct consequence of our main result (Theorem~\ref{thm:equiv_polar_pac}), this semantic-secrecy bound coincides for Polar and PAC codes.

For comparison, we also compute the semantic-secrecy bound of the Invertible Extractor (IE) scheme (Defn.~\ref{defn:secrecycodes-2}) at blocklength $N=256$ using the bound in Theorem~\ref{thm-ie}. This yields a direct finite-length comparison between the code-based secrecy constructions in Defn.~\ref{defn-secrecycodes1} and the extractor-based secrecy construction in Defn.~\ref{defn:secrecycodes-2}. From this analysis, we observe that the secrecy bound under the IE construction \cite{bellare2012polynomial} remains the same for a fixed blocklength, irrespective of the error probabilities. In contrast, the leakage bound for the Polar/PAC secrecy codes \cite{Hessam-Vardy} improves as Eve’s error probability increases.

\begin{remark}
Although the semantic-secrecy bound for the IE scheme at $N=256$ is theoretically computed to be $13$ using Theorem~\ref{thm-ie}, as shown in Table~\ref{tab:performance}, this bound does not account for the reliability constraints imposed by the corresponding construction \cite[Lemma~5.10]{bellare2012polynomial}. In fact, for the parameters considered here, there exists no positive $b(s)$ corresponding to Bob FER $\leq 0.3$ under the construction in \cite[Lemma~5.10]{bellare2012polynomial}. Therefore, a direct operational comparison with the Polar/PAC constructions at $N=256$ would not be meaningful under these parameter settings. To enable a more balanced comparison between the code-based and extractor-based constructions, we additionally consider blocklength $N=512$, where the corresponding analysis is presented in Sec.~\ref{sec: more-simulations}.
\end{remark}
\begin{table}[t]
\centering
\renewcommand{\arraystretch}{1.1}
\setlength{\tabcolsep}{3.5pt}
\small
\begin{tabular}{|c|c|c|c|c|c|c|c|c|}
\hline
\multicolumn{7}{|c|}{} & \multicolumn{2}{c|}{$\delta^{ds}$} \\
\hline
$p_e$ & $k$ & $\bar{I}$ & $k-\bar{I}$ & $C_s$ & $R_s$ & $R_{\mathrm{eff}}$ 
& Polar/PAC  & IE \\
\hline
0.15 & 72  & 58 & 14  & 0.323 & 0.281 & 0.055 & 11 & 13 \\
0.20 & 92  & 46 & 46  & 0.435 & 0.359 & 0.180 &  10 & 13 \\
0.25 & 104 & 33 & 71  & 0.525 & 0.406 & 0.277 &   9 & 13 \\
0.30 & 113 & 24 & 89  & 0.595 & 0.441 & 0.348 &  7 & 13 \\
0.35 & 117 & 13 & 104 & 0.648 & 0.457 & 0.406 &  5 & 13 \\
0.40 & 121 &  7 & 114 & 0.685 & 0.473 & 0.445 &  4 & 13 \\
\hline
\end{tabular}
\caption{Finite-length secrecy evaluation for $\mathsf{BSC}(p_b,p_e)$ at $N=256$ and $p_b=0.05$. 
Here, $p_e$ is Eve's crossover probability, $k$ is the number of message bits, 
$\bar{I}$ is an upper bound on $I(\mathbf{M};\mathbf{Z})$ as per \eqref{eq:secrecy}, 
$R_s = k/N$ is the secrecy rate, and 
$R_{\mathrm{eff}} = (k-\bar{I})/N$ is the effective secrecy rate; $C_s$ is the wiretap channel capacity and $\delta^{ds}$ is the semantic-secrecy bound.\\
\textit{Example:} For completeness, we illustrate how the semantic-secrecy bound for the Polar/PAC construction is computed for one entry in Table~\ref{tab:performance}. Consider the row corresponding to $p_e=0.40$. Using the Tal-Vardy construction with $\mu=64$, Eve's polarized bit-channel capacities are computed, yielding
$
\bar{I}=\sum_{i\in\mathcal{R}^c}C(W_i) \tilde= 7
$ following Theorem~\ref{thm:secrecy}.
The corresponding semantic-secrecy bound is then
$
\delta^{ds}(\text{Polar/PAC})
\leq \sqrt{2\bar{I}}
=\sqrt{2(7)}
~= 3.74,
$
which is rounded to $4$ in Table~\ref{tab:performance}. Further, from Theorem~\ref{thm:equiv_polar_pac}, the same bound applies to both the Polar and PAC constructions. Next, for the IE construction, Theorem~\ref{thm-ie} gives
$
\delta^{ds}(\text{IE})
\leq 6\cdot2^{-\sqrt{N}}
=6\cdot2^{-16}
\tilde= 9.16\times10^{-5}.
$
Equivalently, evaluating
$
-\log_2\left(\delta^{ds}(\text{IE})\right),
$
and expressing in terms of the corresponding security exponent, gives approximately $13$ for $N=256$.}
\label{tab:performance}
\end{table}
\begin{table}[t]
\centering

\renewcommand{\arraystretch}{1.2}
\setlength{\tabcolsep}{8pt}
\begin{tabular}{|c|c|c|c|c|}
\hline
\multicolumn{3}{|c|}{} 
& \multicolumn{2}{|c|}{Information Leakage Bound ($\delta^{ds}$)} \\
\cline{4-5}
\hline
$p_e$ & $b(s)/N$ & $k(s)/N$ & IE & Polar/ PAC \\
\hline
0.391 & 0.0413 & 0.8546 & 21 & 6 \\
0.411 & 0.0529 & 0.8546 & 21 & 5 \\
0.431 & 0.0621 & 0.8546 & 21 & 4 \\
0.451 & 0.0689 & 0.8546 & 21 & 3 \\
0.471 & 0.0735 & 0.8546 & 21 & 2 \\
0.491 & 0.0757 & 0.8546 & 21 & 1 \\
\hline
\end{tabular}
\caption{Finite-length information leakage comparison at $N=512$, $p_b=0.005$. Here, $b(s)$ and $k(s)$ are parameters as defined in \cite[Lemma~5.10]{bellare2012polynomial}.}\label{tab:performance2}
\end{table}

\subsection{More Simulations: Semantic secrecy comparison with the invertible-extractor scheme} \label{sec: more-simulations}
In Table~\ref{tab:performance2}, we compare the finite-length performance of the secrecy constructions defined in Defn.~\ref{defn-secrecycodes1} and Defn.~\ref{defn:secrecycodes-2} at blocklength $N=512$. The main channel crossover probability is fixed to $p_b = 0.005,$ and Bob's frame error rate (FER) is maintained at approximately $0.004$ (equivalently, $\nu(s) = 0.1$ in the notation of \cite[Lemma~5.10]{bellare2012polynomial}.

The IE construction (Defn.~\ref{defn:secrecycodes-2}) is instantiated using the Invert-then-Encode (ItE)-framework as proposed in \cite[Lemma~5.10]{bellare2012polynomial}. We use Polar code for the blockcode $\mathcal{C}$ in the IE construction. For a $\mathsf{BSC}(p_b,p_e)$ channel, the code parameters satisfy $k(s) = (1-h_2(p_b) - \nu(s))N,$ where $k(s)$ is the minimum secret message length, $\nu(s)$ is the bound on Bob's FER, and $N$ is the block length; and, the extractor parameters are chosen as per the construction defined in \cite[Lemma~5.10]{bellare2012polynomial}. Then, the semantic-secrecy guarantee is proven to be upper bounded by $6 . 2^{-\sqrt{N}}.$ At a fixed blocklength $N=512,$ this bound evaluates to a constant that depends only on $N,$ and not explicitly on Eve's crossover probability $p_e.$ Consequently, the IE semantic-secrecy bound remains essentially constant across the tested values of $p_e,$ as reflected in Table~\ref{tab:performance2}.

In contrast, the Polar/PAC secrecy code construction (Defn.~\ref{defn-secrecycodes1}) derives its secrecy guarantee by explicitly upper-bounding the mutual information of Eve's synthesized bit-channels. As $p_e$ increases Eve's channel becomes more degraded, and a larger fraction of her bit-channels fall into the poor set---defined in \cite{Hessam-Vardy}. Since the leakage bound $\bar{I}$ is determined by the aggregate contribution of these bit-channels, the resulting semantic-secrecy bound $\delta^{ds} \leq \sqrt{2\bar{I}}$ tightens as $p_e$ increases.

Hence, we can conclude that while the IE construction \cite{bellare2012polynomial} provides a block length-driven secrecy guarantee that is independent of the specific value of $p_e,$ the Polar and PAC secrecy code constructions \cite{Hessam-Vardy} yield channel-dependent secrecy bounds that improve with increasing degradation of Eve’s channel.

Note that for Polar/PAC the secrecy bounds of Polar/PAC using secrecy coding scheme from Defn.~\ref{defn-secrecycodes1} as shown in Table~\ref{tab:performance} can be further improved by choosing a larger truncation value on the output alphabet size (for instance, $\mu = 256$). Another important observation is that the secrecy guarantees derived from Defn.~\ref{defn-secrecycodes1} are specific to Polar and PAC linear block code constructions. In contrast, the IE-based bounds apply more generally, as they are obtained through the extractor-based framework and do not rely on the specific structural properties of Polar or PAC codes.

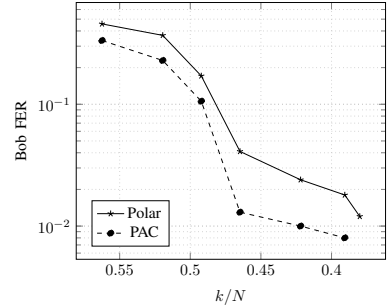
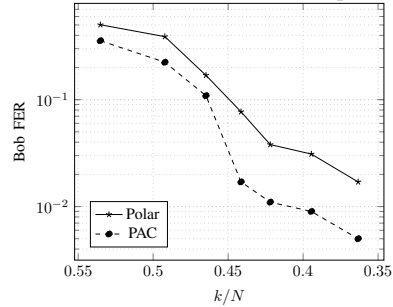
\begin{figure}[ht!]
  \centering
  \begin{subfigure}{0.5\textwidth}
    \centering
     \begin{tikzpicture}[scale=0.6]
        \begin{semilogyaxis}[
            grid=both, 
            grid style=dotted,
            xlabel={$k/N$},
            ylabel={Bob FER}, 
            x dir=reverse, legend entries = {Polar, PAC }, legend style ={at={(0.05,0.05)},anchor=south west}
        ]
        \addplot[color=black,mark=star] coordinates {
            (0.648,0.7895)
            (0.588,0.6038)
            (0.570,0.5714)
            (0.559,0.4313)
            (0.531,0.3413)
            (0.512,0.3226)
            (0.486,0.1596)
            (0.465,0.0314)
            (0.445,0.0143)
            (0.379,0.0072)
            (0.273,0.0070)
        };
        \addplot[color=black,mark=*,dashed] coordinates {
            (0.648,0.7143)
            (0.588,0.4878)
            (0.570,0.4831)
            (0.559,0.3865)
            (0.531,0.3922)
            (0.512,0.2967)
            (0.486,0.1545)
            (0.465,0.0229)
            (0.445,0.0110)
            (0.379,0.0068)
            (0.273,0.0066)
        };

        \end{semilogyaxis}
    \end{tikzpicture}
    \caption{$p_b = 0.05, p_e = 0.4, C_s = 0.6846$ bits per channel use}
    \label{fig:setup1}
  \end{subfigure}
  \begin{subfigure}{0.5\textwidth}
    \centering
    \begin{tikzpicture}[scale=0.6]
        \begin{semilogyaxis}[
            grid=both, 
            grid style=dotted,
            xlabel={$k/N$},
            ylabel={Bob FER}, 
            x dir=reverse, legend entries = {Polar, PAC }, legend style ={at={(0.05,0.05)},anchor=south west}
        ]
        \addplot[color=black,mark=star] coordinates {
            (0.598,0.7407)
            (0.537,0.5755)
            (0.520,0.5882)
            (0.508,0.4372)
            (0.480,0.3096)
            (0.461,0.2967)
            (0.436,0.1582)
            (0.414,0.0336)
            (0.395,0.0115)
            (0.328,0.0090)
            (0.223,0.0074)
        };
        \addplot[color=black,mark=*,dashed] coordinates {
            (0.598,0.7317)
            (0.537,0.5161)
            (0.520,0.5376)
            (0.508,0.3441)
            (0.480,0.2967)
            (0.461,0.3003)
            (0.436,0.1620)
            (0.414,0.0181)
            (0.395,0.0090)
            (0.328,0.0116)
            (0.223,0.0066)
        };
        \end{semilogyaxis}
    \end{tikzpicture}
    \caption{$p_b = 0.05, p_e = 0.3, C_s = 0.5949$ bits per channel use}
    \label{fig:setup2}
  \end{subfigure}
  \caption{Performance of Polar and PAC codes in terms of Bob FER for Wiretap BSC}
  \label{fig:wiretap-polar-pac}
\end{figure}
\vspace{-0.1in}
\subsection{Reliability over the main channel}
In Fig.\ref{fig:wiretap-polar-pac}, we illustrate Bob’s Frame Error Rate (FER) under different wiretap channel configurations. In particular, Fig.\ref{fig:setup1} presents the results for $\mathsf{BSC}(0.05, 0.4)$, where Bob’s FER is plotted for rates near the secrecy capacity $C_s = 0.6846$. Similarly, Fig.~\ref{fig:setup2} considers $\mathsf{BSC}(0.05,0.3)$ with the rates chosen near the corresponding secrecy capacity $C_s = 0.5949.$ All the simulations are conducted at block length of $N = 256$, and the decoding is performed using Successive Cancellation List (SCL) decoding with list size  $L=16$. The decoder implementation is described in Sec.~\ref{sec-scldecoder}.

From Fig.~\ref{fig:wiretap-polar-pac}, we observe that Polar and PAC codes exhibit comparable FER performance across the considered secrecy rates and channel configurations. Thus, under the finite-length settings considered here, the introduction of the convolutional precoder in PAC codes does not result in a significant reliability improvement over Polar codes.

\begin{remark}
   To generate the reliability sequence and to evaluate the mutual information of the bit-channels required for computing the bound in \eqref{eq:secrecy}, we employ the Tal–Vardy construction method \cite{TAL-VARDY}. Specifically, Algorithms~A~and~C from \cite{TAL-VARDY} are used to construct the good set in polar secrecy codes \cite{Hessam-Vardy}, while Algorithm B from \cite{TAL-VARDY} is used to construct the poor set---defined in \cite{Hessam-Vardy}, with the output alphabet size truncated at $\mu = 64$. As said earlier, the secrecy bounds of Polar/PAC using secrecy coding scheme from Defn.~\ref{defn-secrecycodes1} as shown in Table~\ref{tab:performance}~and~\ref{tab:performance2} can be further improved by choosing a larger truncation value on the output alphabet size (such as, $\mu = 256$).
\end{remark}

\section{Discussion and Open Challenges}
Our finite-length experiments support two high-level observations. First, when Polar and PAC codes are used within the same wiretap coset-coding construction, their secrecy guarantees (as evaluated via bit-channel-capacity leakage bounds and converted to semantic secrecy) coincide. Our simulations further indicate that Polar and PAC codes exhibit comparable reliability performance at Bob under SCL decoding for the considered finite-length settings. Second, the IE framework provides semantic-secrecy bounds that are largely independent of Eve's channel parameter $p_e$. Consequently, for a given blocklength, IE bounds can be conservative compared to Polar/PAC coset-coding bounds, which tighten significantly as Eve becomes noisier.

Several open directions remain. Tighter finite-length semantic-secrecy bounds for the IE framework on specific channels (e.g., the wiretap BSC) would enable a sharper comparison at small $N$. For coset codes, designing $(\mathcal{A},\mathcal{R})$ to optimize \emph{semantic secrecy at fixed FER} (rather than asymptotic rate) remains largely empirical; efficient design rules beyond Tal--Vardy style constructions would be valuable. Finally, extending the same comparison methodology to non-symmetric or non-degraded wiretap channels is an important step toward broader applicability.

\appendices
\section{Bit-Channel Equivalence of Polar and PAC Codes}\label{app:equiv}
In this section, we begin with a brief overview of Polar and PAC codes, originally introduced by Arıkan in \cite{arikan-polarcodes,arikan2019sequential}, followed by a discussion of Successive Cancellation List (SCL) decoding, as proposed in \cite{tal2015list}. We then present an optimal definition of bit-channels for PAC codes in the context of secrecy coding, and establish our main theoretical result: these bit-channels are information-theoretically equivalent to those of Polar codes.

\subsection{Polar Codes}\label{sec:polarcodes}
Polar codes, introduced by Arıkan in \cite{arikan-polarcodes}, are a class of linear block codes constructed for block lengths $N = 2^n$. Let $\mathcal{X},$ and $\mathcal{Y}$ denote the binary input and output alphabets, respectively. Then, for a binary symmetric channel (BSC) $\langle \mathcal{X}, \mathcal{Y}, W \rangle$, a polar code is defined by the generator matrix
$$G_N = B_N F^{\otimes n},$$
where $B_N$ is the bit-reversal permutation matrix, $F = \begin{bmatrix} 1 & 0 \\ 1 & 1 \end{bmatrix}$ is the polarization kernel, and $\otimes$ denotes the Kronecker product. An information set $A \subseteq \{1, 2, \ldots, N\}$ consists of the $k$ most reliable bit-channels and is used to carry a $k$-bit message $\mathbf{M}$. The complementary set $A^c = [N] \setminus A$ corresponds to frozen bits, typically fixed to zero.

The underpinning idea behind Polar codes is \emph{channel polarization}, wherein $N$ independent copies of the BSC are recursively combined and split into $N$ synthesized bit-channels $W_1, W_2, \ldots, W_N$. As $N \to \infty$, these synthesized channels polarize, approaching either noiseless or completely noisy extremes. The reliability of each $W_i$ can be quantified using the Bhattacharyya parameter $Z(W_i)$ or the mutual information $I(W_i)$, with smaller $Z(W_i)$ (or larger $I(W_i)$) indicating greater reliability for the bit-channel $i$; see \cite{arikan-polarcodes,TAL-VARDY} for the definitions and more details.

\emph{Polar Encoding:} To encode a $k$-bit message $\mathbf{M} \in \{0,1\}^k$, the encoder constructs an $N$-length vector $\mathbf{V}$ by placing $\mathbf{M}$ in the indices of $A$ and assigning zeros to $A^c$. The codeword is then obtained as 
$$
    \mathbf{X} = \mathbf{V}G_N,
$$
which is transmitted over $N$ independent uses of the BSC, yielding the received sequence $\mathbf{Y} \in \mathcal{Y}^N$.

\subsection{Polarization Adjusted Convolutional Codes}

We now review Polarization-Adjusted Convolutional (PAC) codes, introduced by Arıkan in \cite{arikan2019sequential}. These codes are an extended version of polar codes that incorporates a convolutional precoding stage before the polar transform block (see Fig.~\ref{fig:PACschematic}). Below, we briefly describe their encoding, followed by a discussion on decoding, with an emphasis on the SCL decoding method from \cite{tal2015list,rowshan2021polarization}.

\begin{figure}
    \centering
\begin{tikzpicture}
\draw[-,dashed] (-0.5,1) -- (7.6,1) -- (7.6,-1) -- (-0.5,-1) -- (-0.5,1);
\draw[->] (-0.1,0) -- (0.5,0) node[above,midway] {$\mathbf{M}$};
 \node[draw, minimum width=1cm, minimum height=1cm, align = center] at (1.2,0) (main) {Rate \\ Profiling};
 \draw[->] (1.9,0) -- (2.6,0) node[above,midway] {$\mathbf{U}$};
\node[draw, minimum width=1cm, minimum height=1cm, align = center] at (3.1,0) (main) {$T_N^g$};
\draw[->] (3.6,0) -- (4.4,0) node[above,midway] {$\mathbf{V}$};
\node[draw, minimum width=1cm, minimum height=1cm, align = center] at (4.9,0) (main) {$G_N$};
\draw[->] (5.4,0) -- (6,0) node[above,midway] {$\mathbf{X}$};
\node[draw, minimum width=1cm, minimum height=1cm, align = center] at (6.5,0) (main) {$W_1^N$};
\draw[->] (7,0) -- (7.5,0) node[above,midway] {$\mathbf{Y}$};
\end{tikzpicture}
\caption{Schematic of PAC codes}
    \label{fig:PACschematic}
\end{figure}
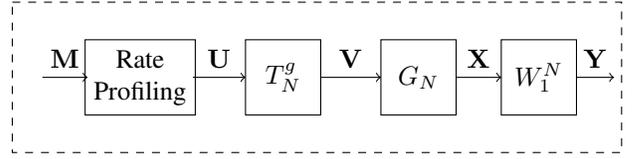

\emph{PAC Encoding:} A $k$-bit message $\mathbf{M}$ is first embedded into an $N$-length vector $\mathbf{U}$ by placing the message bits in the indices of $A$. This vector is processed by a convolutional transform block, followed by the polar transform; see Fig.~\ref{fig:PACschematic}. The resulting codeword is given by
$
\mathbf{X} = \mathbf{U} T_N^g G_N,
$
where $G_N$ is the polar generator matrix defined in sec.~\ref{sec:polarcodes}, and $T^g_N$ is an $N \times N$ Toeplitz matrix determined by the generator polynomial
\begin{align*}
    g(D) = 1+ g_1 D + \ldots + g_{m-1} D^{m-1} + g_m D^m,
\end{align*}
with $g_i \in \{0,1\}$ and memory $m > 1$. Further, note that the matrix $T^g_N$ has an upper-triangular structure, i.e., for instance, when $m=3,$ it is defined as follows:
\begin{align*}
   T^g_N =  
\begin{bmatrix}
        1&g_1&g_2&1  &0  &0  &0  &0  &0  &0\\
        0&1  &g_1&g_2&1  &0  &0  &0  &0  &0\\
        0&0  &1  &g_1&g_2&1  &0  &0  &0  &0\\
        0&0  &0  &1  &g_1&g_2&1  &0  &0  &0\\
        0&0  &0  &0  &1  &g_1&g_2&1  &0  &0\\
        0&0  &0  &0  &0  &1  &g_1&g_2&1  &0\\
        0&0  &0  &0  &0  &0  &1  &g_1&g_2&1\\
        0&0  &0  &0  &0  &0  &0  &1  &g_1&g_2\\
        0&0  &0  &0  &0  &0  &0  &0  &1  &g_1\\
        0&0  &0  &0  &0  &0  &0  &0  &0  &1
        \end{bmatrix}.
\end{align*}

\subsection{Successive Cancellation List Decoding}\label{sec-scldecoder}

Note that both Polar and PAC codes can be decoded using the Successive Cancellation List (SCL) decoding proposed in \cite{tal2015list}. Unlike the SC decoder proposed in \cite{arikan-polarcodes}, the SCL decoder reduces the problem of potentially erroneous decisions in SC decoding by maintaining a multi-path decoder that contains a list of the $L$
most likely decoding paths. In particular, the SCL decoder works based on the following three-step approach.

Let $\hat{\mathbf{V}}$ be the estimate of codewords obtained by passing the received information bits $\mathbf{Y}$ through SCL decoding. Then:
\begin{itemize}
    \item \emph{Step-1}: At information node $i \in A,$ the decoder computes the log-likelihood ratio (LLR):
    $$\lambda_0^i = \ln \left(\frac{P(\mathbf{Y},\hat{V}_1^{i-1}|\hat{V}_i=0)}{P(\mathbf{Y},\hat{V}_1^{i-1}|\hat{V}_i=1)}\right).$$
    Next, instead of making a hard decision as in SC decoding, the decoder branches into two candidate paths, corresponding to the possible values of $\hat{V}_i \in \{0,1\}.$ Further, a heuristic function is also defined as
    $$h(\lambda_0^i) = \begin{cases}
    & 0, ~\text{if}~ \lambda_0^i > 0,\\
    & 1, \text{otherwise}.
\end{cases}$$
\item \emph{Step-2:} For each of the possible paths $l$ obtained through the above branching, the path metric (PM) is updated as
\begin{align*}
    PM_i[l] = \begin{cases}
        & PM_{i-1}[l] + \lvert \lambda_0^i\rvert,~ \text{if~}\hat{V}_i[l] \neq h(\lambda_0^i), \\
        & PM_{i-1}[l],~\text{otherwise.}
    \end{cases}
\end{align*}
Note that at each node $i,$ the branching generates $2L$ paths from previous $L$ paths at node $i-1.$

\item \emph{Step-3:} Finally, at every node $i,$ the decoder retains the $L$ most promising paths---the paths with lowest path metric---out of the above possible 2L paths. At the final node $i=N,$ the path with the minimum path metric is selected as the decoded codeword $\hat{\mathbf{V}}.$
\end{itemize}

In case of PAC codes, the estimate of the information sequence $\hat{\mathbf{U}}$ is obtained from the $\hat{\mathbf{V}}$ through the \emph{known} convolution transform. That is, 
\begin{align*}
    \hat{U}_1 &= \hat{V}_1 \\
    \hat{U}_i &= \hat{V}_i + g_1 \hat{U}_{i-1} + \ldots + g_{m-1} \hat{U}_{i-{m-1}} + g_m \hat{U}_{i-m}.
\end{align*}

\subsection{Bit-Channel Definitions}
We now recall the bit-channel definition for Polar codes and provide an optimal bit-channel definition for PAC codes in the context of secrecy coding. Note that, from \cite{arikan-polarcodes}, we have
$I(\mathbf{V};\mathbf{Y}) = I(\mathbf{X};\mathbf{Y}).$
Further, using chain rule of mutual information, we have
\begin{definition}
    \cite{arikan-polarcodes} The $i^{th}$ `polar bit-channel' $W_i$ is defined as a channel with binary input $V_i$ and output $(\mathbf{Y},V_1^{i-1}),$ for $i=1,\ldots,N$.
\end{definition}
Further, we now show that introducing a convolutional precoding block before the Polar transformation does not alter the overall mutual information. That is,
\begin{theorem}\label{thm-mutuinf}
    The mutual informations $I(\mathbf{U};\mathbf{Y}), I(\mathbf{V};\mathbf{Y}),$ and $I(\mathbf{X};\mathbf{Y})$ are all equivalent.
\end{theorem}
\begin{proof}
    The proof follows directly from the lemma below.
\end{proof}

\begin{lemma}\label{lemma:lemma1}
    If $S \to f(S) \to T$ is a Markov chain and $f(.)$ is a deterministic, one-to-one function, then $I(S;T) = I(f(S);T).$
\end{lemma}
\begin{proof}
    Note that, from the chain rule of mutual information we have,
    \begin{align}\label{eq-mutinf}
        I(S, f(S);T) &= I(S;T) + I(f(S);T|S) \\
        &= I(f(S);T) + I(S;T|f(S)).
    \end{align}
Further, since $S \to f(S) \to T,$ is a Markov chain, we have $I(S;T|f(S)) = 0.$ Additionally, 
$
    I(f(S);T|S) = H(T|S) - H(T|S,f(S)) =0.
$ Substituting these into \eqref{eq-mutinf} completes the proof.
\end{proof}

Next, from Theorem~\ref{thm-mutuinf} and arguments analogous to the Polar case, we can now define the bit-channels of PAC codes as follows:

\begin{definition}
    The $i^{th}$ `PAC bit-channel,' denoted by $\tilde{W}_i,$ is defined as a channel with binary input $U_i$ and output $(\mathbf{Y},U_1^{i-1}),$ for $i=1,\ldots,N$.
\end{definition}
We emphasize that this definition is \emph{optimal}, since it entails no loss of information. More precisely, we have the following:
\begin{theorem}\label{thm2}
    In the notation and setting for PAC codes introduced earlier, we have
    \begin{align*}
        I(U_i|Y_1^N,U_1^{i-1}) = I(V_i|Y_1^N, V_1^{i-1}), ~i=1,2,\ldots,N.
    \end{align*}
\end{theorem}
\begin{proof}
    This result follows directly from the equivalence of Polar and PAC bit-channels (as defined above), which we rigorously establish in the next section. An alternative proof is provided in Appendix~\ref{alternate-proof}.
\end{proof}

\subsection{Main Result: Equivalence of Polar and PAC Bit Channels}\label{Polar-PAC}

Following the above definitions, we now establish our key theoretical result that \emph{the Polar and PAC bit-channels are equivalent}. Specifically, we show that the channels $W_i: V_i \to (\mathbf{Y}, V_1^{i-1})$ and $\tilde{W}_i: U_i \to (\mathbf{Y}, U_1^{i-1})$ are equivalent in the sense of channel symmetry. To proceed, we begin with the following definition.

\begin{definition}
    A $2\times a$ matrix $A$ is said to be symmetric if there exists a permutation $\pi$ of $\{1,\ldots,a\}$ with $\pi^{-1} = \pi$ and $A(0,i) = A(1,\pi(i))$ for all $i \in \{1,\ldots,a\}$. 
    \label{def:sym}
\end{definition}
It follows that $A(0,\pi(i)) = A(1,\pi(\pi(i))) = A(1,i)$. Hence, whenever $i \neq \pi(i)$, Columns $i$ and $\pi(i)$ are flipped versions of each other, while if $i = \pi(i)$, the two entries in Column $i$ are identical. Consequently, exchanging the two rows of a symmetric matrix produces a column-permuted version of the original matrix.

With this notion in hand, we can now proceed to compare the bit-channels of Polar and PAC codes. Note that both the channels $W_i$ and $\tilde{W}_i$ are of size $ 2 \times |\mathcal{Y}|^N2^{i-1}.$ Let the columns of $W_i$ are indexed by $(y_1^N,v_1^{i-1}),$ with $y_1^N \in \mathcal{Y}^N,$ and $v_1^{i-1} \in \{0,1\}^{i-1}.$ Likewise, the columns of $\tilde{W}_i$ are indexed by $(y_1^N,u_1^{i-1}),$ with $y_1^N \in \mathcal{Y}^N,$ and $u_1^{i-1} \in \{0,1\}^{i-1}.$ Further, for each $i$, let $W_i(v_1^{i-1})$ denote the $2 \times |\mathcal{Y}|^N$ submatrix of $W_i$ obtained by fixing $v_1^{i-1}\in\{0,1\}^{i-1}$ and collecting the columns indexed by $(y_1^N, v_1^{i-1})$ with $y_1^N \in \mathcal{Y}^N$. The corresponding submatrices $\tilde{W}_i(u_1^{i-1})$ for PAC codes are defined in the same way. Then, the following lemma is an important observation.

\begin{lemma}
The submatrix $W_i(v_1^{i-1})$ is symmetric for each $v_1^{i-1}\in\{0,1\}^{i-1}$.
\label{lem:sym}
\end{lemma}
\begin{proof}
See sec.~\ref{proof:lem:sym} for the proof.
\end{proof}
Next, we present our main result.
\begin{theorem}
    The bit channels $W_i$ of the polar code and the bit channel $\tilde{W}_i$ of the PAC code (as defined above) are equivalent by a column permutation of their matrices for every $i$.
\end{theorem}
\begin{proof}
Note that for a given $v_1^i$, we can determine the corresponding input to the convolutional code $u_1^i(v_1^i)$ by the $1-1$ relationship
$$v_1^i=u_1^i(v_1^i)\,\left[T_N^g\right]_{1:i,1:i},$$
where $\left[T_N^g\right]_{a:b,c:d}$ denotes the submatrix of $T_N^g$ containing the rows $a,a+1,\ldots,b$ and columns $c,c+1,\ldots,d$. To reduce clutter, we denote $u_1^i(v_1^i)$ as $u_1^i,$ and $P(Y=y|X=x)$ as $P(y|x)$ wherever necessary.

For a fixed $v_1^i$ and the corresponding $u_1^i$, there is a $1-1$ relationship between $v_{i+1}^N$ and $u_{i+1}^N$ given by
\begin{align*}
    &[v_1^i\quad v_{i+1}^N] \\ &= [u_1^i\quad u_{i+1}^N]T_N^g \\
    &= \left[u_1^i\left[T_N^g\right]_{1:i,1:i}\quad u_1^i\left[T_N^g\right]_{1:i,i+1:N}+u_{i+1}^N\left[T_N^g\right]_{i+1:N,i+1:N}\right].
\end{align*}

Using the above and $v_1^N=u_1^NT_g$, we get
\begin{align*}
    &W_i(Y_1^N=y_1^N,V_1^{i-1}=v_1^{i-1}|V_i=v_i)\nonumber\\
     &=\sum_{v_{i+1}^N\in\{0,1\}^{N-i}}P(y_1^N,v_1^{i-1},v_{i+1}^N|v_i)\nonumber\\
     &=\sum_{v_{i+1}^N\in\{0,1\}^{N-i}}P(y_1^N|v_1^N)P(v_1^{i-1},v_{i+1}^N)\nonumber\\
    &=\sum_{u_{i+1}^N\in\{0,1\}^{N-i}}P(y_1^N|u_1^N)P(u_1^{i-1},u_{i+1}^N)\nonumber\\
     \end{align*}
    \begin{align}
    &=\sum_{u_{i+1}^N\in\{0,1\}^{N-i}}P(y_1^N|u_1^{i-1},u_{i+1}^N,u_i)P(u_1^{i-1},u_{i+1}^N)\nonumber\\
    &=\sum_{u_{i+1}^N\in\{0,1\}^{N-i}}P(y_1^N,u_1^{i-1},u_{i+1}^N|u_i)\nonumber\\
    &=P(Y_1^N=y_1^N,U_1^{i-1}=u_1^{i-1}|U_i=u_i) \nonumber \\
    &=\tilde{W}_i(Y_1^N=y_1^N,U_1^{i-1}=u_1^{i-1}|U_i=u_i).
\end{align}
Additionally, we have that
$$v_i=u_i\oplus g_1u_{i-1}\oplus\cdots\oplus g_{m-1}u_{i-(m-1)}\oplus u_{i-m}.$$
So, for a fixed $v_1^{i-1}$ and the corresponding $u_1^{i-1}$, the submatrix $W_i(v_1^{i-1})$ of the polar bit channel has a correspondence with the submatrix $\tilde{W}_i(u_1^{i-1})$ of the PAC bit channel. The correspondence is as follows: $\tilde{W}_i(u_1^{i-1}) = W_i(v_1^{i-1})$ if $g_1u_{i-1}\oplus\cdots\oplus g_{m-1}u_{i-(m-1)}\oplus u_{i-m}=0,$ or $\tilde{W}_i(u_1^{i-1})$ is the row-exchange of $W_i(v_1^{i-1})$ if $g_1u_{i-1}\oplus\cdots\oplus g_{m-1}u_{i-(m-1)}\oplus u_{i-m}=1.$

Since $W_i(v_1^{i-1})$ is symmetric by Lemma \ref{lem:sym}, for every $v_1^{i-1}$ and the corresponding $u_1^{i-1}$, we have that
\begin{equation}
    \tilde{W}_i(u_1^{i-1}) = \text{Column-permute}(W_i(v_1^{i-1})).
\end{equation}
This implies that the matrix of $W_i$ is a column-permuted version of the matrix of $\tilde{W}_i$ and the proof is complete.    
\end{proof}

\subsection{Proof of Lemma~\ref{lem:sym}}\label{proof:lem:sym}
Recall that $W$ is a symmetric channel with input $\{0,1\}$ and output $\mathcal{Y}$. Also, $N=2^n$ and, for $i=1,\ldots,N$, the $i$-th bit channel of the polar code $W_i$ has input $\{0,1\}$ and output $\mathcal{Y}^N\times\{0,1\}^{i-1}$. We will denote $W_i$ as $W^{(n)}_i$ to bring out the dependence on the exponent of blocklength, i.e., $n,$ explicitly.

The claim of the lemma is that $W^{(n)}_i(v_1^{i-1})$ (a submatrix of $W^{(n)}_i$ with outputs $(y_1^N,v_1^{i-1})$ for all $y_1^N\in \mathcal{Y}^N$ and fixed $v_1^{i-1}$) is symmetric (as per Definition \ref{def:sym}) for every $v_1^{i-1}\in\{0,1\}^{i-1}$. The proof will be by induction on $n$. 

The base case $n=0$ results in $W^{(0)}_1=W$, which is symmetric by assumption.

As the induction assumption, we will assume that $W^{(n)}_i(v_1^{i-1})$ is symmetric for every $i$ and $v_1^{i-1}$. We will denote the symmetry permutation of $W^{(n)}_i(v_1^{i-1})$ as $\pi_{v_1^{i-1}}$. So, for $s_1^n\in\mathcal{Y}^N$, we have that
\begin{align}
    W^{(n)}_i(s_1^N,v_1^{i-1}|0)=W^{(n)}_i(\pi_{v_1^N}(s_1^N),v_1^{i-1}|1),\label{eq:sym1}\\
    W^{(n)}_i(s_1^N,v_1^{i-1}|1)=W^{(n)}_i(\pi_{v_1^N}(s_1^N),u_1^{i-1}|0).\label{eq:sym2}    
\end{align}
When clear, we will use the shorthand notation $\overline{s_1^N}\triangleq \pi_{u_1^N}(s_1^N)$. By the definition of the polar code, we have the following recursions for the bit channels from blocklength exponent $n$ to $n+1$.
\begin{align}
    & W^{(n+1)}_{2i-2}(y_1,y_2|v_1) \\
    &\triangleq \frac{1}{2}\left(W_i^{(n)}(y_1|v_1)W_i^{(n)}(y_2|0)+W_i^{(n)}(y_1|v_1\oplus1)W_i^{(n)}(y_2|1)\right),\nonumber\\
    & W^{(n+1)}_{2i-1}(y_1,y_2,v_1|v_2) \\ &\triangleq \frac{1}{2}W_i^{(n)}(y_1|v_1\oplus v_2)W_i^{(n)}(y_2|v_2)\nonumber
\end{align}
for $y_1=(s_1^N,a_1^{i-1})\in\mathcal{Y}^N\times\{0,1\}^{i-1}$, $y_2=(t_1^N,b_1^{i-1})\in\mathcal{Y}^N\times\{0,1\}^{i-1}$, $v_1\in\{0,1\}$, $v_2\in\{0,1\}$. 
So, we have
\begin{equation}
    \begin{split}
        &W^{(n+1)}_{2i-2}(s_1^N,a_1^{i-1},t_1^N,b_1^{i-1}|0)\\
&= \frac{1}{2}\Big(
W_i^{(n)}(s_1^N,a_1^{i-1}|0)\,
W_i^{(n)}(t_1^N,b_1^{i-1}|0) \nonumber\\
&\quad +
W_i^{(n)}(s_1^N,a_1^{i-1}|1)\,
W_i^{(n)}(t_1^N,b_1^{i-1}|1)
\Big)\nonumber\\
&\overset{(a)}{=} \frac{1}{2}\Big(
W_i^{(n)}(\overline{s_1^N},a_1^{i-1}|1)\,
W_i^{(n)}(t_1^N,b_1^{i-1}|0) \nonumber\\
&\quad +
W_i^{(n)}(\overline{s_1^N},a_1^{i-1}|0)\,
W_i^{(n)}(t_1^N,b_1^{i-1}|1)
\Big)\nonumber\\
&= W^{(n+1)}_{2i-2}(\overline{s_1^N},a_1^{i-1},t_1^N,b_1^{i-1}|1).
    \end{split}
\end{equation}
where $(a)$ follows by using \eqref{eq:sym1} and \eqref{eq:sym2}. Similarly, we can show that 
$$W^{(n+1)}_{2i-2}(s_1^N,a_1^{i-1},t_1^N,b_1^{i-1}|1)=W^{(n+1)}_{2i-2}(\overline{s_1^N},a_1^{i-1},t_1^N,b_1^{i-1}|0).$$
This shows that $W^{(n+1)}_{2i-2}(a_1^{i-1},b_1^{i-1})$ is symmetric for every $a_1^{i-1},b_1^{i-1}$.

Proceeding similarly for $W^{(n+1)}_{2i-1}$, we see that
\begin{align}
    & W^{(n+1)}_{2i-1}(s_1^N,a_1^{i-1},t_1^N,b_1^{i-1},v_1|0)\\
    &=\frac{1}{2}W_i^{(n)}(s_1^N,a_1^{i-1}|v_1)W_i^{(n)}(t_1^N,b_1^{i-1}|0)\nonumber\\
    &\overset{(a)}{=} \frac{1}{2}W_i^{(n)}(\overline{s_1^N},a_1^{i-1}|v_1\oplus1)W_i^{(n)}(\overline{t_1^N},b_1^{i-1}|1)\nonumber\\
    &=W^{(n+1)}_{2i-1}(\overline{s_1^N},a_1^{i-1},\overline{t_1^N},b_1^{i-1},v_1|1),\nonumber
\end{align}
where $(a)$ follows by using \eqref{eq:sym1} and \eqref{eq:sym2}. Note that $\overline{s_1^N}=\pi_{a_1^{i-1}}(s_1^N)$ and $\overline{t_1^N}=\pi_{b_1^{i-1}}(t_1^N)$. Similarly, we can show that
\begin{align*}
    &W^{(n+1)}_{2i-1}(s_1^N,a_1^{i-1},t_1^N,b_1^{i-1},v_1|1) \\
    &= W^{(n+1)}_{2i-1}(\overline{s_1^N},a_1^{i-1},\overline{t_1^N},b_1^{i-1},v_1|0).
\end{align*}

This shows that $W^{(n+1)}_{2i-1}(a_1^{i-1},b_1^{i-1},v_1)$ is symmetric for every $a_1^{i-1},b_1^{i-1},v_1$. This concludes the induction step and the proof.                  \qed

\section{Alternative Proof of Theorem~\ref{thm2}}\label{alternate-proof}
Firstly, note that
\begin{align*}
    I(U_i;Y_1^N,X_1^{i-1}) &= I(U_i;U_1^{i-1}) + I(U_i;Y_1^N|X_1^{i-1}) \\
    I(U_i;Y_1^N|U_1^{i-1}),
\end{align*}
because $U_i$ is independent of $U_1^{i-1}$. Similarly, since $V_i$ is independent of $V_1^{i-1},$ we have that
\begin{align*}
    I(V_i;Y_1^N,V_1^{i-1}) = I(V_i;Y_1^N|V_1^{i-1}).
\end{align*}
Consequently, we will equivalently show by induction that
\begin{align}\label{eq:goal}
    I(U_i;Y_1^N|U_1^{i-1}) = I(V_i;Y_1^N|V_1^{i-1}),
\end{align}
for $i= 1,2,\ldots,N.$ The base case is $i=1,$ which is true because $U_1 = V_1.$ As the induction hypothesis, assume that \eqref{eq:goal} is true for $i=1,\ldots,j.$ Using Lemma~\ref{lemma:lemma1} with $S = U_1^{j+1},$ and $T=Y_1^N,$ we get

\begin{align*}
    I(U_1^{j+1};Y_1^N) &= I(V_1^{j+1};Y_1^N) \\
    \sum_{i=1}^{j+1} I(U_i;Y_1^N|U_1^{i-1}) &= \sum_{i=1}^{j+1} I(V_i;Y_1^N|V_1^{i-1}) \\
    I(U_{j+1};Y_1^N|U_1^j) &= I(V_{j+1};Y_1^N|V_1^j),
\end{align*}
where the first step uses the chain rule and the second step uses the induction hypothesis, completing the induction step and proof. \qed
\vspace{-0.1in}
\section{Semantic security of wiretap coset codes}
Consider an $(n,n-k)$ linear binary code with generator matrix $G$. Let $G'$ be a $k\times n$ binary matrix such that the concatenated $n\times n$ matrix $\begin{bmatrix}G\\G'\end{bmatrix}$ has rank $n$. Under wiretap coset coding, a $k$-bit message $m$ is encoded to a codeword $c=mG'+\mathbf{U}G$, where $\mathbf{U}$ is a uniformly random $(n-k)$-bit vector. Let the wiretap channel be a binary-input, symmetric channel $W$ with output alphabet $\mathcal{Y}$. The codeword $c=[c_1,\ldots,c_n]$ is sent over $W$ and an output vector $\mathbf{Z}$ is received. Since $W$ is symmetric, there exists a permutation $\pi$ on $\mathcal{Y}$ such that $W(y|0)=W(\pi(y)|1)$ for every $y\in\mathcal{Y}$ and $\pi^{-1}=\pi$. For bits $a,b$, we have that 
$$W(y|a+b)=W(\pi^b(y)|a),\text{ where }\pi^b(y)=\begin{cases}
    y,&b=0,\\
    \pi(y),&b=1.
\end{cases}$$
Let us denote by $W^*$ the overall channel with input $m=[m_1,\ldots,m_k]$ and output $\mathbf{Z}$, i.e. the channel with input alphabet $\{0,1\}^k$ and output alphabet $\mathcal{Y}^n$ with transition probability
\begin{align}
&W^*(y_1^n|m)=P([y_1,\ldots,y_n]\,\vert\,[m_1,\ldots,m_k])\nonumber\\
&=\sum_{\substack{u\in\{0,1\}^{n-k}\\c=mG'+uG}}\frac{1}{2^{n-k}}W(y_1|c_1)\cdots W(y_n|c_n).
\end{align}
The claim is that $W^*$ is a symmetric channel. This claim implies that $I(\mathbf{M};\mathbf{Z})$ is maximized when $\mathbf{M}$ is uniformly random over $\{0,1\}^k$ and shows the equivalence of strong secrecy and semantic secrecy.

To show $W^*$ is symmetric, consider the partition of $\mathcal{Y}^n$ (the outputs of $W^*$) defined through the following equivalence relation: for $y_1^n,{y'}_1^n\in\mathcal{Y}^n$,
\begin{gather*}
y_1^n\sim {y'}_1^n \text{ if }
[y'_1,\ldots,y'_n]=[\pi^{v_1}(y_1),\ldots,\pi^{v_n}(y_n)] \text{ for }\\
v=[v_1,\ldots,v_n]\in\text{rowspace}(G').
\end{gather*}
Within each partition defined by $\sim$ above, we will show that the rows and columns of the transition matrix are permutations of each other. For $m\in\{0,1\}^k$, $v=mG'$, let us consider the two columns indexed by outputs $[y_1,\ldots,y_n]$ and $[\pi^{v_1}(y_1),\ldots,\pi^{v_n}(y_n)]$ in one partition. For $m'\in\{0,1\}^k$, we see that 
\begin{align*}
&W^*([y_1,\ldots,y_n]|m')\\
&=\sum_{\substack{u\in\{0,1\}^{n-k}\\c=m'G'+uG}}\frac{1}{2^{n-k}}W(y_1|c_1)\cdots W(y_n|c_n)\\
&=\sum_{\substack{u\in\{0,1\}^{n-k}\\c=(m'+m)G'+uG+v\\c'=(m'+m)G'+uG}}\frac{1}{2^{n-k}}W(y_1|c'_1+v_1)\cdots W(y_n|c'_n+v_n)\\
&=\sum_{u\in\{0,1\}^{n-k}}\frac{1}{2^{n-k}}W(\pi^{v_1}(y_1)|c'_1)\cdots W(\pi^{v_n}(y_n)|c'_n)\\
&=W^*([\pi^{v_1}(y_1),\ldots,\pi^{v_n}(y_n)]|m+m').
\end{align*}
As $m'$ varies over $\{0,1\}^k$, the column in the partition indexed by $[y_1,\ldots,y_n]$ is seen to be a permutation of the column indexed by $[\pi^{v_1}(y_1),\ldots,\pi^{v_n}(y_n)]$ in the same partition.

Similarly, considering the two rows indexed by $m$ and $\tilde{m}$, we can show that
$$W^*([y_1,\ldots,y_n]|m)=W^*([\pi^{v_1}(y_1),\ldots,\pi^{v_n}(y_n)]|\tilde{m})$$
for $v=(m'+m+\tilde{m})G'$, $m'\in\{0,1\}^k$. As $m'$ varies over $\{0,1\}^k$, the row indexed by $m$ is seen to be a permutation of the row indexed by $\tilde{m}$ within the partition containing $[y_1,\ldots,y_n]$.

\vspace{-0.1in}
\bibliography{References} 

@ARTICLE{1055763,
  author={Leung-Yan-Cheong, S.},
  journal={IEEE Transactions on Information Theory}, 
  title={On a special class of wiretap channels (Corresp.)}, 
  year={1977},
  volume={23},
  number={5},
  pages={625-627},
  doi={10.1109/TIT.1977.1055763}}

@ARTICLE{arikan-polarcodes,
  author={Arikan, Erdal},
  journal={IEEE Transactions on Information Theory}, 
  title={Channel Polarization: A Method for Constructing Capacity-Achieving Codes for Symmetric Binary-Input Memoryless Channels}, 
  year={2009},
  volume={55},
  number={7},
  pages={3051-3073},
  doi={10.1109/TIT.2009.2021379}}

@ARTICLE{Hessam-Vardy,
  author={Mahdavifar, Hessam and Vardy, Alexander},
  journal={IEEE Transactions on Information Theory}, 
  title={Achieving the Secrecy Capacity of Wiretap Channels Using Polar Codes}, 
  year={2011},
  volume={57},
  number={10},
  pages={6428-6443},
  doi={10.1109/TIT.2011.2162275}}

@ARTICLE{TAL-VARDY,
  author={Tal, Ido and Vardy, Alexander},
  journal={IEEE Transactions on Information Theory}, 
  title={How to Construct Polar Codes}, 
  year={2013},
  volume={59},
  number={10},
  pages={6562-6582},
  doi={10.1109/TIT.2013.2272694}}

@article{arikan2019sequential,
  title={From sequential decoding to channel polarization and back again},
  author={Ar{\i}kan, Erdal},
  journal={arXiv preprint arXiv:1908.09594},
  year={2019}
}

@article{rowshan2021polarization,
  title={Polarization-adjusted convolutional (PAC) codes: Sequential decoding vs list decoding},
  author={Rowshan, Mohammad and Burg, Andreas and Viterbo, Emanuele},
  journal={IEEE Transactions on Vehicular Technology},
  volume={70},
  number={2},
  pages={1434--1447},
  year={2021},
  publisher={IEEE}
}

@article{bellare2012polynomial,
  title={Polynomial-time, semantically-secure encryption achieving the secrecy capacity},
  author={Bellare, Mihir and Tessaro, Stefano},
  journal={arXiv preprint arXiv:1201.3160},
  year={2012}
}

@article{bellare2012cryptographic,
  title={A cryptographic treatment of the wiretap channel},
  author={Bellare, Mihir and Tessaro, Stefano and Vardy, Alexander},
  journal={arXiv preprint arXiv:1201.2205},
  year={2012}
}

@article{tal2015list,
  title={List decoding of polar codes},
  author={Tal, Ido and Vardy, Alexander},
  journal={IEEE transactions on information theory},
  volume={61},
  number={5},
  pages={2213--2226},
  year={2015},
  publisher={IEEE}
}

@article{wyner1975wire,
  title={The wire-tap channel},
  author={Wyner, Aaron D},
  journal={Bell system technical journal},
  volume={54},
  number={8},
  pages={1355--1387},
  year={1975},
  publisher={Wiley Online Library}
}
\bibliographystyle{ieeetr}

\end{document}